%% file: arxiv.tex
\documentclass[journal,twoside,web]{ieeecolor}
\usepackage{lcsys}
\usepackage{xspace}
\input{preamble}

\usepackage{textcomp}
\def\BibTeX{{\rm B\kern-.05em{\sc i\kern-.025em b}\kern-.08em
    T\kern-.1667em\lower.7ex\hbox{E}\kern-.125emX}}
\let\debugmode\undefined 
\def\debugmode{} 

\ifdefined\debugmode
\definecolor{rev_color}{RGB}{214, 0, 110}

\newcommand{\zheinline}[1]{{\color{red} (Zhe: #1)}}
\newcommand{\zhe}[1]{\marginpar{\hspace{-10pt}\color{red}\tiny\ttfamily Zhe: #1}}
\newcommand{\necinline}[1]{{\color{green} (Nec: #1)}}
\newcommand{\nec}[1]
{\marginpar{\hspace{-10pt}\color{green}\tiny\ttfamily Nec: #1}}

\newcommand{\MF}{MOP\xspace}

\else

\definecolor{newContent}{RGB}{0, 0, 0}

\newcommand{\zheinline}[1]{}
\newcommand{\zhe}[1]{}
\newcommand{\necinline}[1]{}
\newcommand{\nec}[1]{}

\fi

\newcommand{\tsfm}{\texttt{TF}}

\begin{document}
\title{
Can Transformers Learn Optimal Filtering for Unknown Systems?
}

\author{
    Haldun Balim \hspace{2em} Zhe Du \hspace{2em} Samet Oymak \hspace{2em} Necmiye Ozay
    \thanks{Manuscript received August 15, 2023. N. Ozay and H. Balim were funded by the ONR grant N00014-21-1-2431 (CLEVR-AI) and NSF Grant CNS 1931982. Z. Du and S. Oymak were funded by the Army Research Office grant W911NF2110312. N. Ozay and S. Oymak were also funded by a MIDAS PODS grant.}
    \thanks{First two authors contributed equally.}
    \thanks{Z. Du is with Department of Mechanical Engineering, University of California, Riverside, CA 92521 USA (e-mail: zhedu@umich.edu).}
    \thanks{H. Balim is with Department of Mechanical and Process Engineering, ETH Zürich, 8092 Zürich, Switzerland (e-mail: hbalim@ethz.ch)}
    \thanks{S. Oymak and N. Ozay are with Electrical Engineering and Computer Science Department, University of Michigan, Ann Arbor, MI 48109 USA (e-mail: \{oymak, necmiye\}@umich.edu) }
}
  
\maketitle
\thispagestyle{empty}

\begin{abstract}
Transformer models have shown great success in natural language processing; however, their potential remains mostly unexplored for dynamical systems. In this work, we investigate the optimal output estimation problem using transformers, which generate output predictions using all the past ones. Particularly, we train the transformer using various distinct systems and then evaluate the performance on unseen systems with unknown dynamics. Empirically, the trained transformer adapts exceedingly well to different unseen systems and even matches the optimal performance given by the Kalman filter for linear systems. In more complex settings with non-i.i.d. noise, time-varying dynamics, and nonlinear dynamics like a quadrotor system with unknown parameters, transformers also demonstrate promising results.  To support our experimental findings, we provide statistical guarantees that quantify the amount of training data required for the transformer to achieve a desired excess risk. Finally, we point out some limitations by identifying two classes of problems that lead to degraded performance, highlighting the need for caution when using transformers for control and estimation.
\end{abstract}

\begin{IEEEkeywords}
Filtering, Neural networks, Statistical learning
\end{IEEEkeywords}

\section{Introduction}
\label{sec_intro}
\IEEEPARstart{M}{any} control problems such as model predictive control and safety analysis are built upon predictions of system's future trajectories. This prediction (or estimation) problem is well studied and dates back to the classical Kalman filter~\cite{re1960new}, which is optimal for linear systems with Gaussian noise. Methods are also developed for more complex setups, e.g. extended Kalman filter \cite{anderson2012optimal} for nonlinear systems, particle filters \cite{del1997nonlinear} when system dynamics can be sampled, and adaptive filters
and adaptive filters \cite{diniz1997adaptive} 
for unknown systems. Existing methods typically require the knowledge of system dynamics, linearity, time-invariance, or Gaussian noise, which, for more challenging and realistic settings, may yield degraded performance. 

Prediction, on the other hand, in the domain of natural language processing, has witnessed recent success thanks to the transformer models \cite{vaswani2017attention}, which are deep learning architectures that can generate text prediction after feeding into an input text sequence. In this work, we investigate the use of transformers in predicting dynamical system's outputs. 

To begin with, we assume a priori access to a collection of $M$ systems drawn from some distribution $\Dcal_{sys}$ and their respective output trajectories $\{\vy_t \}$. These are referred to as source systems and  trajectories respectively. We then train a transformer using the source trajectories so that after feeding into past outputs $\vy_{0:t-1}$, the transformer is able to produce an estimate $\hat{\vy}_{t}$ of the true output $\vy_{t}$ (as in Fig~\ref{fig_illustration}). During test-time, given a previously unseen system from the same distribution $\Dcal_{sys}$, we feed its observed trajectory to the trained transformer and evaluate its prediction performance. As discussed in \cite{li2023transformers}, in this setting transformer acts like a data-driven adaptive algorithm: given a system, the transformer is able to automatically adapt to it and make predictions by leveraging past data. In the remainder of this paper, we refer to a transformer trained in this way as meta-output-predictor~(\MF). 

\begin{figure}
    \centering    \includegraphics[width=0.9\columnwidth]{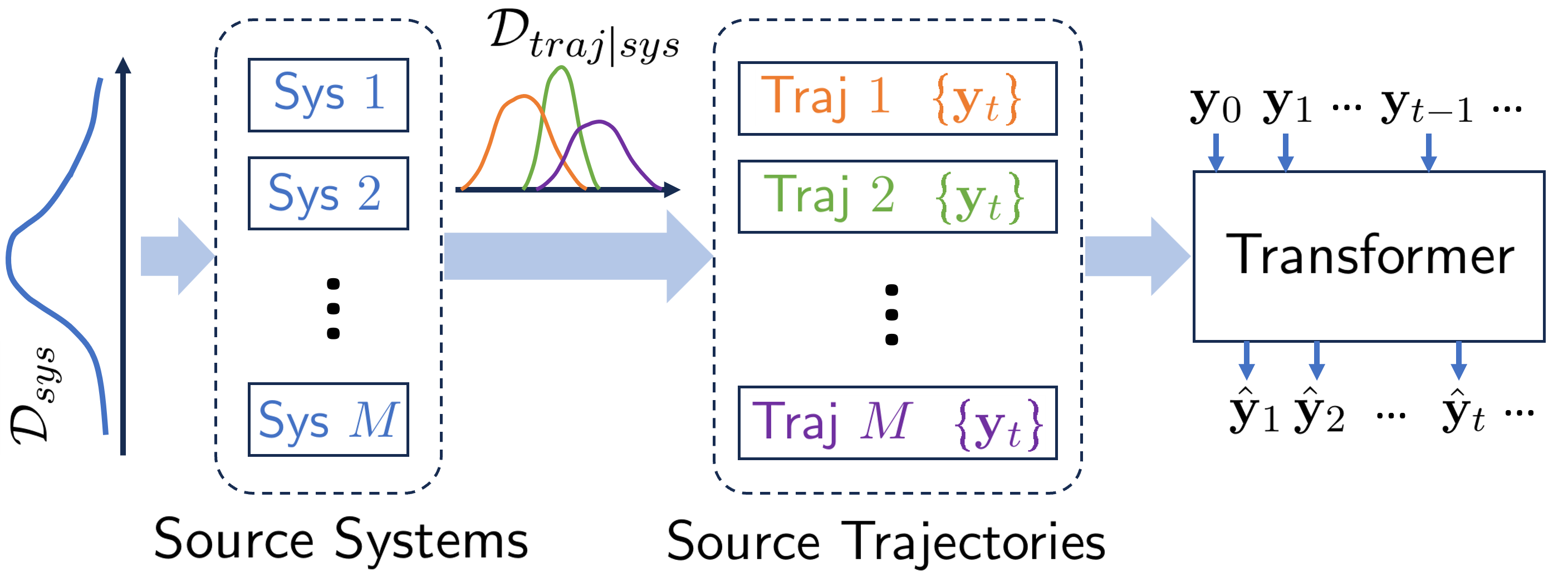}
    \caption{Training a transformer for dynamical system prediction}
    \label{fig_illustration}
    \vspace{-0.4cm}
\end{figure}

\noindent \textbf{Contributions:} 
Our first contribution is numerically demonstrating the capabilities of \MF. The experiments show that \MF matches the optimal performance given by the Kalman filter for different unseen linear systems and is able to handle challenging settings such as non-i.i.d. noise, time-varying dynamics, and nonlinear quadrator systems. Complementing our empirical contributions, we theoretically establish that the excess risk incurred by \MF decays with rate $\Ocal(1/\sqrt{MT})$ where $T$ is the prediction time horizon, under appropriate assumptions. Motivated by our theoretical analysis, we identify a class of systems with slow mixing properties, for which \MF encounters difficulties in learning the optimal estimator. Our experiments also indicate some limitations of \MF in the presence of distribution shifts.

\noindent \textbf{Related Work}:
Compared with earlier neural sequence models, transformers \cite{vaswani2017attention} incorporate the \emph{attention} mechanism that is able to better keep longer memories thus can handle longer input sequences. As a result, a transformer can be trained to perform a variety of tasks rather than a single task \cite{brown2020language, liu2021makes, zhao2021calibrate}, which is known as in-context learning and serves as the foundation of \MF training in our work. Particularly, transformers are shown to be able to in-context learn linear functions \cite{garg2022can}; in-context reinforcement learning is studied in \cite{laskin2022context}. Recent work in~\cite{li2023transformers} studies theoretical properties of transformer-based in-context learning for both i.i.d. data and data with Markovian temporal dependencies (i.e., state trajectories), and provides guarantees in terms of excess risk and transfer risk. Compared with \cite{li2023transformers}, we (i) consider the system output prediction problem with data being non-Markovian, (ii) demonstrate the versatility of \MF through evaluations on several challenging scenarios, and (iii) study scenarios that can lead to degradations in \MF performance.

In terms of filtering/prediction for dynamical systems, there have been many recent advances. When the system dynamics is known, observer design for deterministic systems is studied in \cite{yi2022reduced} through  contraction analysis.
On the other hand, data-driven adaptive methods have received growing attention. For nonlinear system, techniques such as kernel methods \cite{CHEN201695} and nonlinear splines \cite{SCARPINITI2013772} are studied.
Linear system setups allow for more principled methods such as online optimization \cite{anava2013online}, explicit \cite{tsiamis2020sample} or implicit \cite{kozdoba2019line, tsiamis2022online} system identification, and policy optimization \cite{umenberger2022globally}. Given a class of systems, existing works typically propose algorithms, through which a predictor/filter is learned from data for a specific system. This is in contrast to the framework in our work: 
Training \MF with various source systems in a class empower \MF the generalizability to the whole class.
In other words, the learned \MF is not a specific filter, but a prediction algorithm that can filter any system in the class. And as long as the source systems are representative for the system class, the transformer performance is guaranteed, which is no longer confined by common prerequisites such as dynamics linearity, noise Gaussianity, etc.

\section{Problem Setup}\label{sec_problemSetup}
To solve the output estimation problem for an unknown system, we will train a transformer model with data trajectories generated by the following $M$ source systems $\crlbkt{\Scal_i}_{i =1}^M$ drawn from the same distribution $\Dcal_{sys}$:
\begin{equation}\label{eq_contextModel}
 \Scal_i :
 \left\{
 \begin{aligned}
\vx_{i,t+1} &= f_i(\vx_{i,t}) + \vw_{i,t+1} \\
\vy_{i,t} &= g_i(\vx_{i,t}) + \vv_{i,t},
\end{aligned} \right.
\end{equation}
where $\vx_{i,t} \in \dm{n}{1}$ and $\vy_{i,t} \in \dm{m}{1}$ are the state and output at time $t$ in the $i$th system; $f_i(\cdot)$ and $g_i(\cdot)$ are the state dynamics and output functions with $f_i(0) = 0$ and $g_i(0) = 0$; $\vw_{i,t} \sim \N(0, \sigma_{\vw, i}^2 \vI_n)$ and $\vv_{i,t} \sim \N(0, \sigma_{\vv, i}^2 \vI_m)$ are the process and output noise, which are mutually independent for all $i$ and $t$. For simplicity, it is assumed that the initial state $\vx_{i,0} = 0$. These source systems may be obtained through pre-existing datasets or simulation environments.
The target system under evaluation is denoted by $\Scal_0$, which is drawn from the same distribution $\Dcal_{sys}$ and does not have to be contained within the source systems.

We assume that there exists a constant $L_g > 0$ such that for any $i$, $\vx, \vx'$, $\norm{g_i(\vx) - g_i(\vx')} \leq L_g \norm{\vx - \vx'}$. 
Let $\sigma_\vw {:=} \max_i \sigma_{\vw,i}$ and $\sigma_\vv {:=} \max_i \sigma_{\vv,i}$. 
Furthermore, we assume these systems satisfy the following stability condition.
\begin{assumption}[Stability]\label{asmp_stability}
	Let $f_i^{(t)}(\cdot, \cdot)$ denote the $t$-step state evolution function such that $f_i^{(t)}(\vx_{i,\tau}, \vw_{i, \tau+1:\tau+t}) = \vx_{i,\tau + t}$ for all $\tau \in \mathbb{N}$. Then, there exists constants $\rho \in [0,1)$ and $C_\rho>0$ such that for any system $i$ and time step $t$, for any $\vx, \vx'$ and noise sequence $\Wcal:=\crlbkt{\vw_{(\tau)}}_{\tau \in [t]}$, we have
	\begin{equation}
    \label{eq_stability}
		\norm{f_i^{(t)}(\vx, \Wcal) - f_i^{(t)}(\vx', \Wcal)} \leq C_\rho \rho^t \norm{\vx - \vx'}.
	\end{equation}
\end{assumption}

We define the notation $L_\rho := \frac{C_\rho}{1 - \rho}$. When the class of dynamical systems we are sampling from consists of linear systems with $f_i(\vx) = \vA_i \vx$, then Assumption~\ref{asmp_stability} is satisfied when the spectral radius $\rho(\vA_i)<1$ for all $i$. It is also satisfied by systems that are contracting \cite{LOHMILLER1998683} or exponentially incrementally input-to-state stable \cite{Angeli2002alyapunov} with input $\vw$.

In this work, we seek to predict system output using a transformer \cite{vaswani2017attention}, which is a deep sequence model~$\tsfm_\theta(\cdot)$ that maps system output sequences $\Ycal_t := \vy_{0:t}$ to $\vyhat_{t+1}:= \tsfm_\theta(\Ycal_t)$, an estimation of the true output $\vy_{t+1}$ at time $t+1$. The trainable parameters of the transformer are denoted by $\theta \in \Theta$ for some parameter set $\Theta$. The transformer structure allows the sequence length $t$ to be varying. 

Assuming the access to $M$ length-$T$ output trajectories $\{ \vy_{i,0:T} \}_{i = 1}^M$  generated by each of the $M$ source systems, the goal in this work is to train a transformer model that, at each time $t$, can predict the output $\vy_{0,t+1}$ of the target system $S_0$ only using the past outputs $\vy_{0,0:t}$.
Let $\Ycal_{i, t}:= \vy_{i, 0:t}$ denote the outputs up to time $t$, which is also known as the \emph{prompt} (to predict $\vy_{i,t+1}$), then the transformer is trained by solving
\begin{equation}
\label{eq_incontextTraining}
	\widehat{\tsfm} = \underset{\tsfm \in \Acal}{\arg\min} \frac{1}{M T} \sum_{i=1}^M \sum_{t=1}^T \ell(\vy_{i, t}, \tsfm(\Ycal_{i, t-1})),
\end{equation}
where $\Acal := \{ \tsfm_\theta : \theta \in \Theta \}$ and $\ell(\cdot, \cdot) \geq 0$ is the loss function. To apply $\widehat{\tsfm}(\cdot)$ to the target systems $\Scal_0$, we simply take $\widehat{\tsfm}(\Ycal_{0, t})$ as the prediction for $\vy_{0,t+1}$.

Training a model as in \eqref{eq_incontextTraining} where the data comes from a diversity of sources is also known as in-context learning. As a result of the training diversity, the transformer can achieve good performance on any of the source systems as well as demonstrate generalization ability for the unseen target system $\Scal_0$. Hence, we refer to the obtained transformer $\widehat{\tsfm}$ as meta-output-predictor (\MF).

In what follows, we first empirically demonstrate the performance of \MF  in Section~\ref{sec_experiments} under various setups, which is followed by theoretical analysis in Section~\ref{sec_theory}.

\section{Experiments}
\label{sec_experiments}

\begin{figure*}[ht]
    \centering
    \includegraphics[width=0.67\columnwidth]{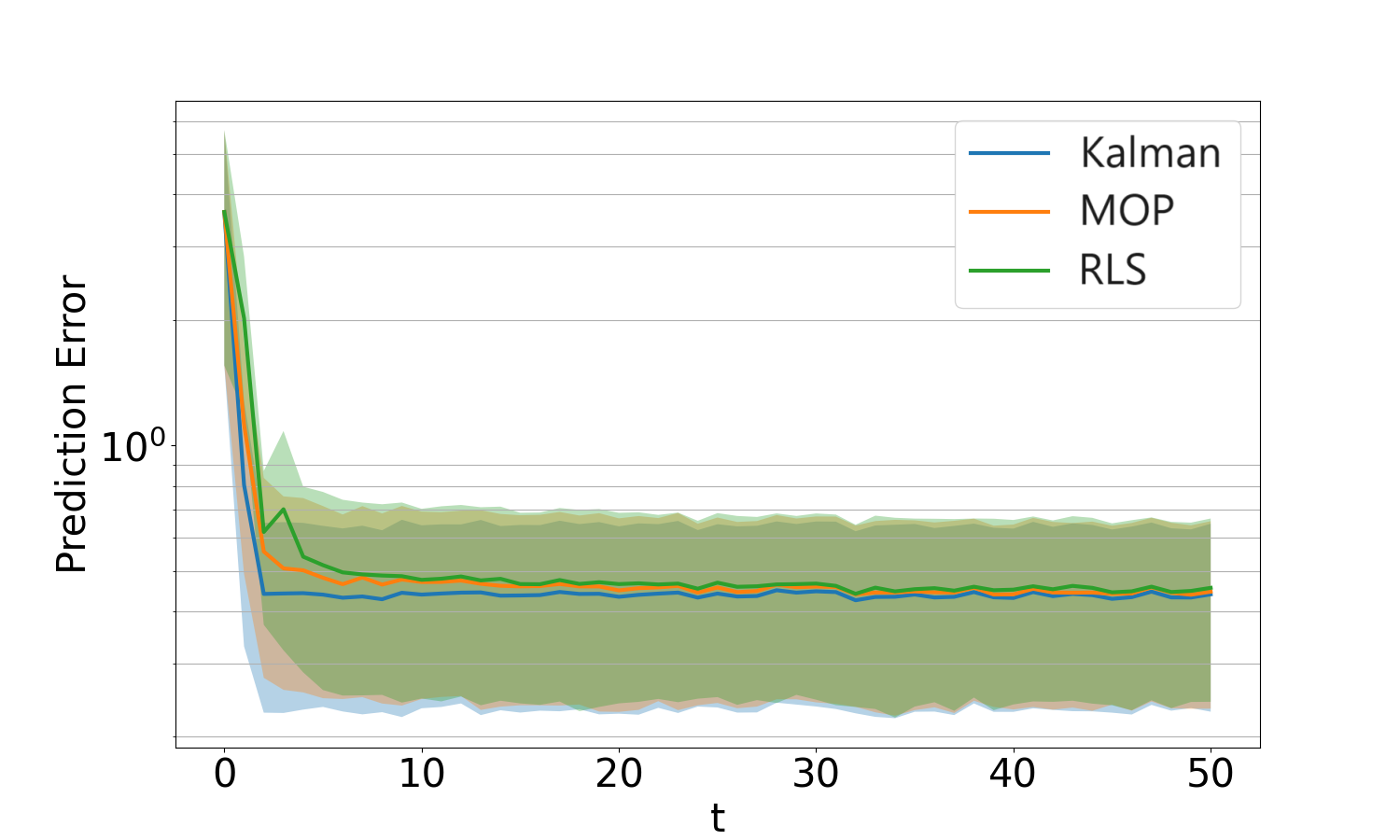}
    \includegraphics[width=0.67\columnwidth]{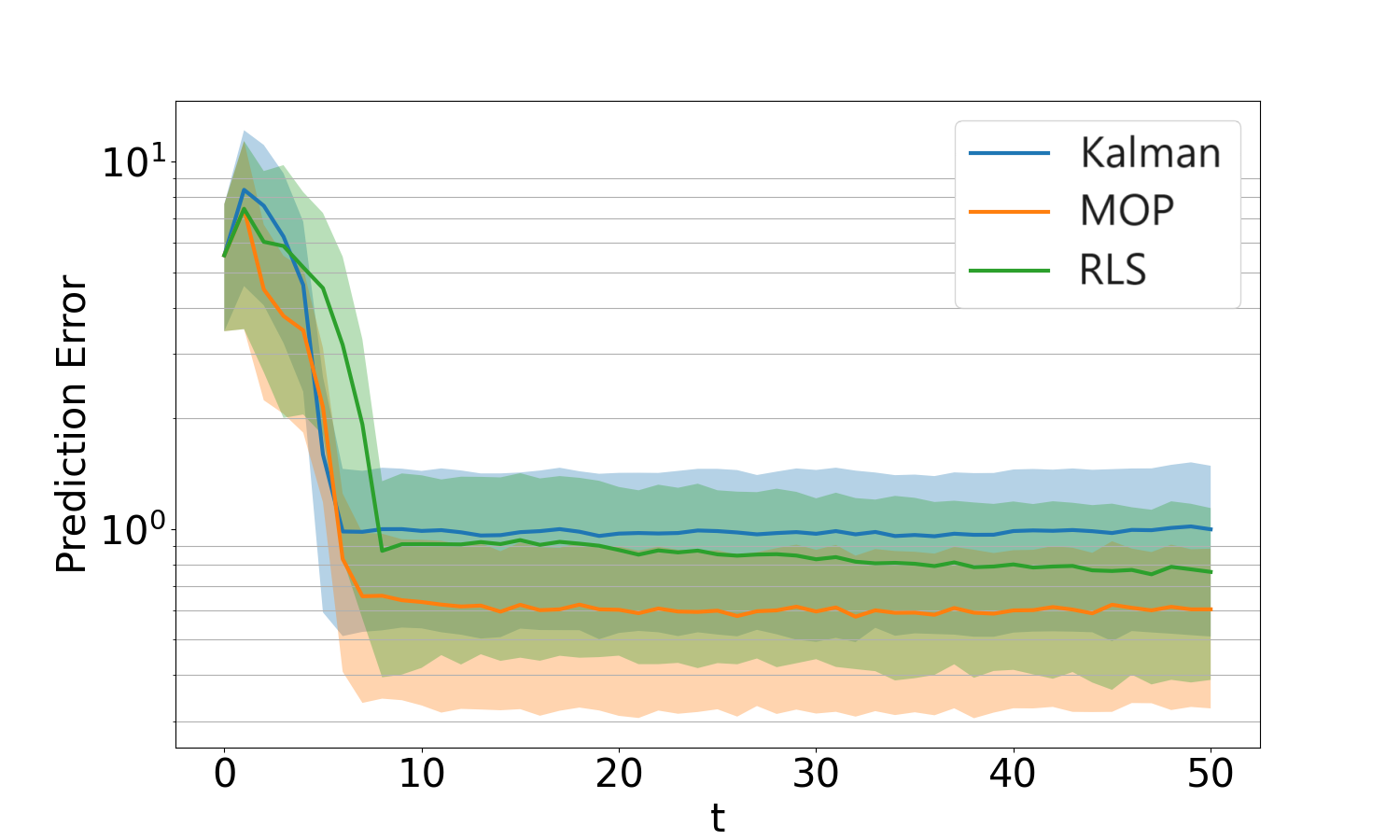}
    \includegraphics[width=0.67\columnwidth]{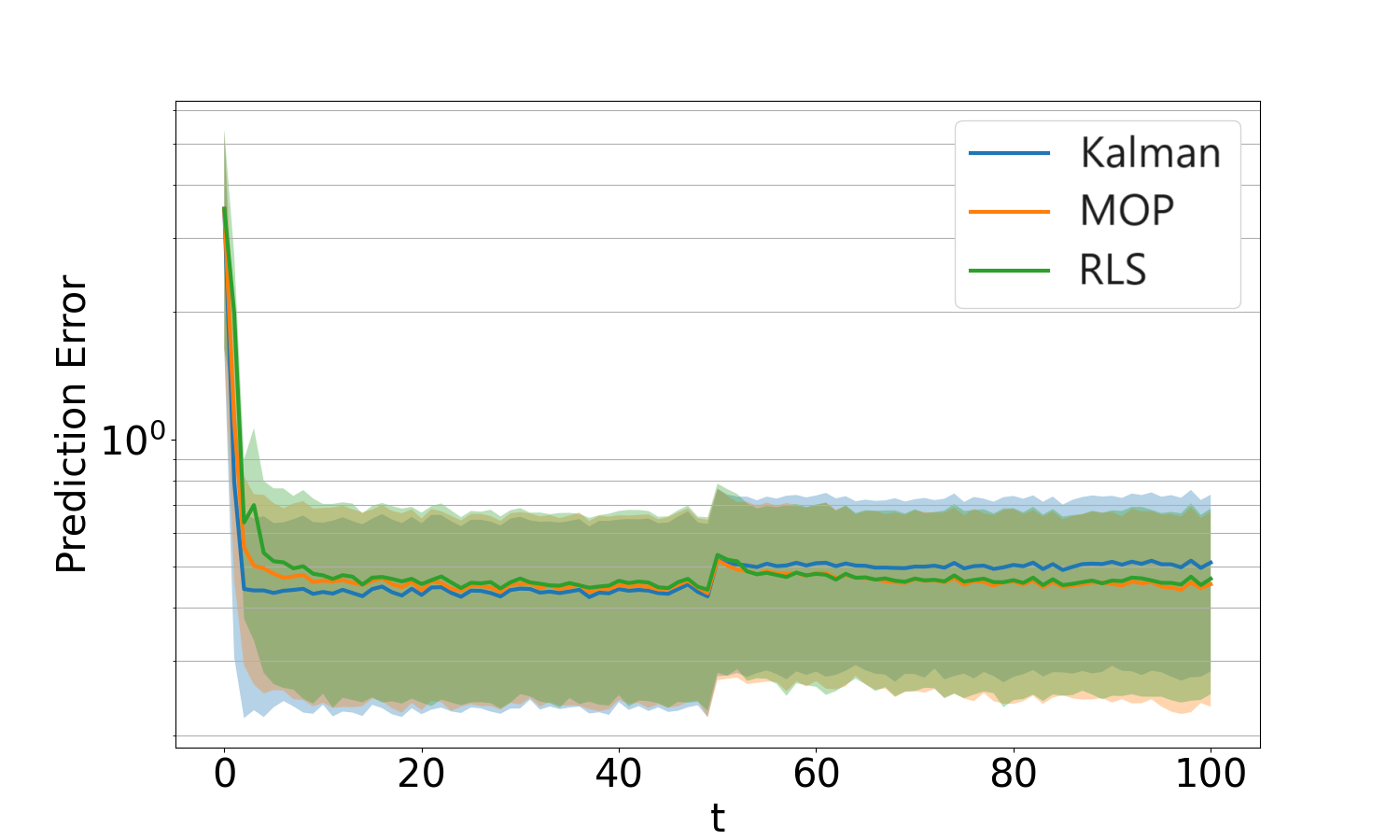}\\    
    \footnotesize \customlabelfig{fig_linSys}{\ref{fig_linear}(a)} (a) \hspace{0.3 \linewidth}
    \customlabelfig{fig_noniid}{\ref{fig_linear}(b)} (b)
    \hspace{0.3 \linewidth}
    \customlabelfig{fig_changingsys}{\ref{fig_linear}(c)} (c)
    \vspace{-1em}
    \caption{Output predictions for linear systems: (a) with i.i.d. Gaussian noise; (b) with colored (non-i.i.d.) noise; (c) with dynamics changes at $t = T/2$. \MF performs as well as or better than Kalman filter even though it does not have access to the system dynamics.}
    \label{fig_linear}
    \vspace{-0.4cm}
\end{figure*}

In this section, we present the experimental results for the transformer-based \MF in different scenarios. In each scenario, during the training, we fix the number of source systems $M=20000$ and training trajectory length $T = 50$. To evaluate the performance of \MF on different unseen test systems, for each experimental setup, we randomly generate $N = 1000$ systems and record the prediction error $\norm{\vyhat_t - \vy_t}$ over trajectories each with length $T = 50$, where $\vyhat_t$ denotes the prediction for $\vy_t$. We use GPT-2 \cite{radford2019language} architecture with 12 layers, 8 attention heads and 256 embedding dimensions. In each experimental setup, the transformer model is trained for $10000$ training steps with batch size $64$. The $\ell_2$-norm is selected as the training loss function. The code we use to produce the figures and execute our algorithm can be accessed at~\url{https://github.com/haldunbalim/Meta-Output-Predictor}

\subsection{Linear Systems} \label{subsec:lin}
We first consider the simplest setting with linear systems and i.i.d. Gaussian noise, i.e., $f(\vx) = \vA \vx$ and $g(\vx) = \vC \vx$ in \eqref{eq_contextModel}. The state dimension is $n = 10$ and the output dimension is $m = 5$. For each source and test system, we generate matrix $\vA$ with entries sampled uniformly between $[0,1]$, which is then followed by scaling so that the largest eigenvalue is $0.95$. The $\vC$ matrix is generated with entries sampled uniformly between $[0,1]$. The noise covariance are $\sigma_{\vw}^2 = 0.01$ and $\sigma_{\vv}^2 =0.01$. Kalman filter and linear autoregressive predictor are used as baselines, where the latter is given by $\vyhat_{t+1} = \valpha_1 \vy_{t} + \valpha_2 \vy_{t-1}$ and the matrix parameters $(\valpha_1, \valpha_2)$ are updated in an online fashion using the recursive least squares (RLS) with initial covariance taken to be identity. This essentially amounts to solving a regularized least squares problem. The results are presented in Fig. \ref{fig_linSys}. We see that after some burn-in time ($\sim$ 20 steps), \MF eventually matches the performance of Kalman filter. This is because the transformer needs to collect certain amount $(\Ocal(n+m))$ of data to implicitly learn the system dynamics, while Kalman filter, designed with the exact system knowledge, reaches optimality immediately.

In the next experiment, we consider the case where the noise process $\crlbkt{\vw_t}$ is non-i.i.d. Specifically, we let $\vw_t = \sum_{t^\prime=t-4}^{t}\veta_{t^\prime}$ and $\vv_t = \sum_{t^\prime=t-4}^{t}\vgamma_{t^\prime}$ where $\veta_{t} \overset{i.i.d.}{\sim} \N(0, 0.01)$ and $\vgamma_{t} \overset{i.i.d.}{\sim} \N(0, 0.01)$. When applying the Kalman filter in this case, we disregard the fact that $\vw_t$ and $\vv_t$ each are temporally correlated and simply use  the variances of $\vw_t$ and $\vv_t$ for prediction.
We note that for non-i.i.d. noise Kalman filter is no longer optimal. Fig.~\ref{fig_noniid} shows the results for this case. We can observe the advantage of \MF over Kalman filter as Kalman filter has lost its optimality whereas \MF has learned the non-i.i.d noise prior during training.

Next, we evaluate the ability of \MF to adapt to run-time changes in the dynamics. Specifically, when generating the test trajectories, we change the underlying dynamics to a randomly generated new one at time $t = T/2 = 50$. The results are presented in Fig. \ref{fig_changingsys}. We see that when dynamics changes occur, there are sudden jumps in prediction error for both \MF and the Kalman filter; as we collect more data from the new dynamics, \MF quickly adapts, and achieves the same performance as before at around $t = 100$. The convergence of \MF after dynamics changes is much slower than the one at the beginning because the prompt always contains data from the original system.

\subsection{Planar Quadrotor Systems}

We consider the underactuated 6D planar quadrotor systems as in \cite{singh2019robust} with the following discrete-time dynamics:
\begin{equation*}\small{
\begin{aligned}
    \underbrace{\begin{bmatrix} x_{t+1} \\ z_{t+1} \\  \phi_{t+1} \\ \dot x_{t+1} \\ \dot z_{t+1} \\ \dot \phi_{t+1}\end{bmatrix}}_{=: \vx_{t+1}} 
    &= 
    \begin{bmatrix} x_t + (\dot x_t \cos(\phi_t) - \dot z_t \sin(\phi_t)) \tau \\ 
    z_t + (\dot x_t \sin(\phi_t) + \dot z_t \cos(\phi_t)) \tau\\
    \phi_t + \dot \phi_t \tau \\
    \dot x_t + (\dot z_t \dot \phi_t - g \sin(\phi_t))\tau \\
    \dot z_t + (-\dot x_t \dot \phi_t - g \cos(\phi_t) + (u_{0_t}+u_{1_t})/m)\tau \\
    (u_{0_t}-u_{1_t})l\tau/J
    \end{bmatrix}
    +  w_{t} \\
    y_t &= \vC \vx_t + v_t.
\end{aligned}}
\end{equation*} 
The mass, length and moment of inertia parameters $(m,l,J)$ are chosen uniformly from $[0.5,2]$, $g$ is set to be constant 10. For each system a trajectory is generated by randomly sampled actions. The noise $w,v$ are sampled from $N(0,0.01)$. The discretization time $\tau=0.1$. The matrix $\vC \in \dm{3}{6}$ has elements uniformly sampled in $[0,1]$. The results are provided in Fig.~\ref{fig_drone}. We see that \MF significantly outperforms the extended Kalman filter.

\begin{figure}
    \centering    \includegraphics[width=0.7\columnwidth]{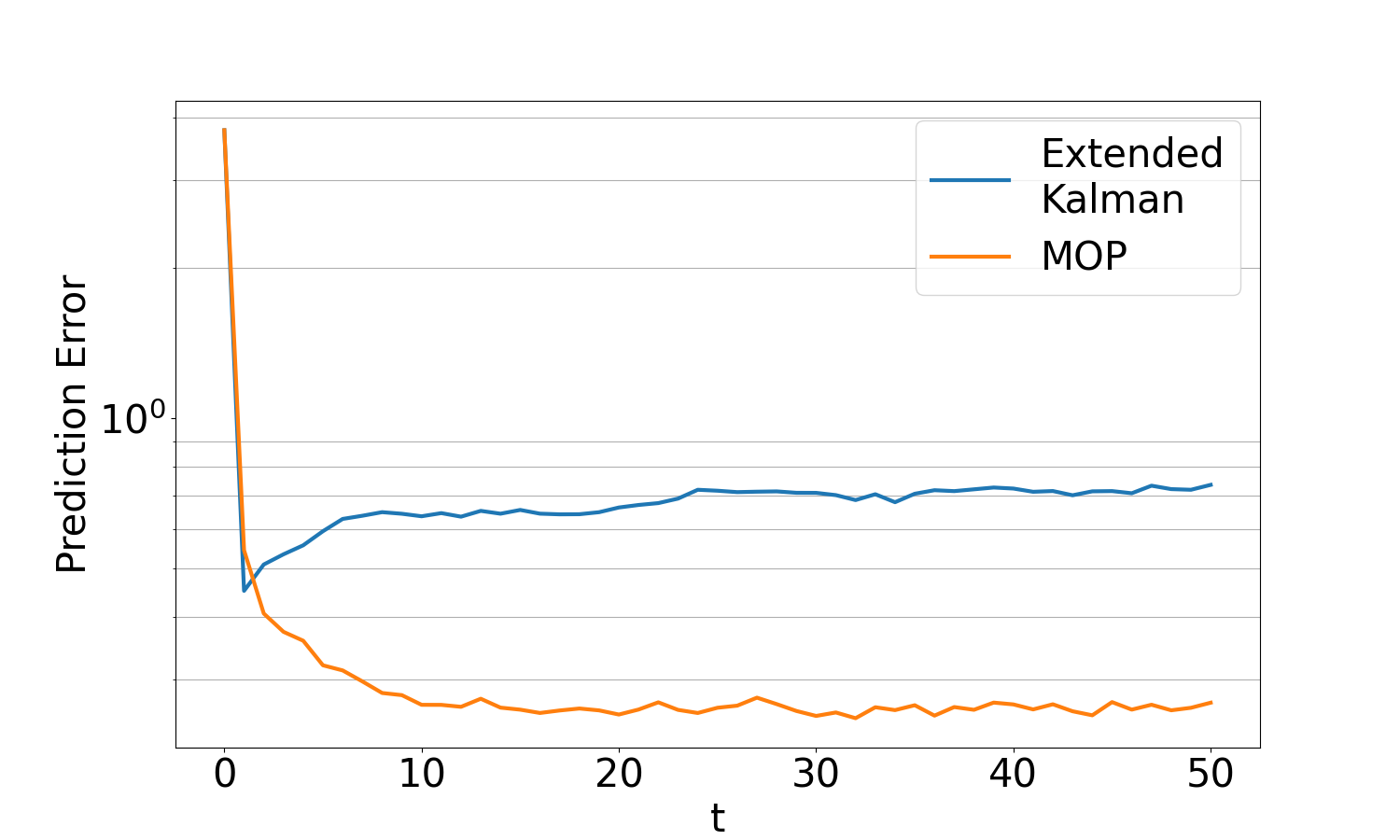}
    \vspace{-1em}
    \caption{Output predictions for planar quadrotor systems.}
    \label{fig_drone}
    \vspace{-0.4cm}
\end{figure}

\section{Theoretical Guarantees}
\label{sec_theory}
Before analyzing the performance of \MF $\widehat{\tsfm}$, we first introduce a few notions and assumptions. The analysis in this section generalizes that in \cite{li2023transformers}, which studies a special case where state is observed (i.e.,  $g$ is known and equal to the identity map and there is no measurement noise).
\subsection{Preliminaries}
\begin{definition}[Covering Number]
\label{def_coveringNum}
	Consider a set $\Qcal$ and a distance metric $d(\cdot, \cdot)$ on $\Qcal$.
	For a set $\bar{\Qcal}_N := \crlbkt{q_1,\dots,q_N}$, we say it is an $\epsilon$-cover of $\Qcal$ if for any $q \in \Qcal$, there exists $q_i \in \bar{\Qcal}_N$ such that $d(q, q_i) \leq \epsilon$.
	The number $\N(\Qcal, d, \epsilon)$ is the smallest $N \in \mathbb{N}$ such that $\bar{\Qcal}_N$ is an $\epsilon$-cover of $\Qcal$.
\end{definition}

We will analyze the set $\Acal$ through its $\epsilon$-cover. To do so, we define the following distance on $\Acal$.

\begin{definition}[Distance Metric]
	\label{def_algDistance}
    Let $\Ycal_T := \vy_{0:T}$ denote a 
     trajectory of some system  
    $i$ under the noise sequence $\crlbkt{\vw_{0:T}, \vv_{0:T}}$.
    For any two transformers $\tsfm, \tsfm' \in \Acal$, define the distance metric 
    $\mu(\tsfm, \tsfm') := \sup_{t \leq T} \sup_{\vw_{0:t}, \vv_{0:t}}
    \frac{\norm{\tsfm(\Ycal_{t-1}) - \tsfm'(\Ycal_{t-1})}}{
    \max_{\tau \leq t} \norm{\vw_{\tau-1}} + \max_{\tau \leq t} \norm{\vv_{\tau-1}}
    }$.
\end{definition}

Though this metric is regarding the transformers $\tsfm, \tsfm'$ in the transformer space $\Acal$, it can be viewed as a metric between their respective parameters $\theta, \theta'$ in the parameter set $\Theta$. 
When the noise is bounded, the distance can also be defined without using the normalization factor in the denominator, e.g.~\cite{li2023transformers}.
Next, we quantify the robustness of the transformer in terms of how much the  prediction changes with respect to the perturbations of its prompt. This will help us establish generalization bounds for the transformer trained in \eqref{eq_incontextTraining}.

\begin{assumption}[Transformer Robustness]
\label{asmp_tsfmRobustness}
	Consider a trajectory $\crlbkt{\vy_{0:t}}$ generated by some system $i$ under the noise sequence $\crlbkt{\vw_{0:t}, \vv_{0:t}}$. Let $\crlbkt{ \vy'_{0:t}}$ denote another trajectory under the same noise except that at time $\tau$, $\crlbkt{\vw_\tau, \vv_\tau}$ is replaced by  $\crlbkt{\vw_\tau', \vv_\tau'}$. 
    Let $\Ycal_{t} := \vy_{0:t}$ and $\Ycal'_{t} := \vy'_{0:t}$.
    Suppose the loss function $\ell(\vy, \cdot)$ is $L_\ell$-Lipschitz.
    Let $\Bcal:= \cap_{j=0}^t \crlbkt{\norm{\vw_j} \leq \wbar, \norm{\vv_j} \leq \vbar}$ for some $\wbar, \vbar \geq 0$ and $\expctn'[\cdot] := \expctn[\cdot \mid \Bcal]$.    
    Then, there exist constants $K \geq 0$ such that for any system $i$, any $t$ and $\crlbkt{\vw_{0:t-1}, \vv_{0:t-1}}$, any $\tau<t$ and $\crlbkt{\vw_\tau', \vv_\tau'}$, and any $\tsfm \in \Acal$, we have
	\begin{multline*}
        \expctn'_{\vw_{t}, \vv_{t}}\big[
		\big| \ell(\vy_{t}, \tsfm(\Ycal_{t-1})) - \ell(\vy_{t}, \tsfm(\Ycal'_{t-1}))
		\big| 	\big] \\ \leq
		\frac{K L_\ell}{t-\tau} \sum_{j = \tau}^{t-1} \norm{\vy_{j} - \vy'_{j}}.
	\end{multline*}
\end{assumption}

In this assumption, trajectories $\vy'_j = \vy_j$ for $j<\tau$ and possibly differ afterward due to the perturbation at time $\tau$, which explains the summation term in the upper bound. It is shown in \cite[Lemma B.5]{li2023transformers} that this assumption holds for a wide class of transformers.

\subsection{Performance Guarantees}

For a transformer $\tsfm \in \Acal$, we define the following risk to evaluate its performance on the target system $S_0$ over the time horizon $T$
\begin{equation}
\label{eq_risk}
	\Lcal(\tsfm) := \frac{1}{T} \sum_{t = 1} ^T\expctn \sqbktbig{ \ell(\vy_{0,t}, \tsfm(\Ycal_{0,t-1})) },
\end{equation}
where the expectation is over the target system $S_0$ and noise terms $\crlbkt{\vw_{0, t}, \vv_{0, t}}$. Let $\tsfm^\star \in \Acal$ denote an optimal transformer that minimizes $\Lcal(\tsfm)$. 
Define the excess risk for $\widehat{\tsfm}$ obtained via minimizing the loss in \eqref{eq_incontextTraining} as
\begin{equation}\label{eq_excessRisk}
	\text{Risk}(\widehat{\tsfm}) := \Lcal(\widehat{\tsfm}) - \Lcal(\tsfm^\star).
\end{equation}
Then, we have the following performance guarantees on  $\widehat{\tsfm}$.
\begin{theorem}
\label{thrm_excessRisk}
	Suppose Assumptions~\ref{asmp_stability} and \ref{asmp_tsfmRobustness} hold, and the loss function $\ell(\vy, \cdot)$ is $L_\ell$-Lipschitz and $\ell(\cdot, \cdot) \leq B$ for some $B \geq 0$. Then, when $MT \geq 3 \max(\sqrt{n}, \sqrt{m})$, for all $\epsilon>0$, with probability at least $1-\delta$,
    \begin{equation*}
        \text{Risk}(\widehat{\tsfm}) 
        \leq 12B\delta + 4 L_\ell \epsilon + \Bbar \sqrt{\frac{\log(4 \N(\Acal, \mu, \epsilon') /\delta)}{c MT}},
    \end{equation*}
	where $\Bbar{:=} 2B + 7 K L_\ell \prnthbig{L_g L_\rho \sigma_\vw {+} \sigma_\vv } \sqrt{\log(4MT/\delta)}\log(T)$; $\epsilon' := \epsilon/((\sigma_\vw + \sigma_\vv) \sqrt{\log(4MT/\delta)})$; $c$ is some absolute constant; $\N(\cdot, \cdot, \cdot)$ and $\mu$ are the covering number and distance metric from Definitions~\ref{def_coveringNum} and \ref{def_algDistance}.
\end{theorem}
 
For fixed failure probability $\delta$ and distance $\epsilon$, the upper bound decays with rate $\Ocal(1/\sqrt{MT})$. When the transformer mapping is Lipschitz, the covering number term can be upper bounded by $\log \N(\Acal, \mu, \epsilon') \leq \Ocal(n_{\Theta} \log(\bar{\Theta} \sqrt{n_{\Theta}} / \epsilon'))$, where $n_{\Theta}$ and $\bar{\Theta}$ are respectively the dimension and magnitude of the transformer parameter set $\Theta$.

\subsection{Proof of the Main Theorem}
In this section, we provide the proof for Theorem~\ref{thrm_excessRisk}. 
Extending Assumption~\ref{asmp_tsfmRobustness}, the following lemma tells how the noise $\{\vw, \vv\}$ would affect the loss performance.
\begin{lemma}
\label{lemma_tsfmRobustness_new}
    Suppose Assumptions~\ref{asmp_stability} and \ref{asmp_tsfmRobustness} hold. Under the same setup as in Assumption~\ref{asmp_tsfmRobustness}, let 
    $\ybar := L_g L_\rho \wbar + \vbar$. Then,
	\begin{equation*}
		\expctn'_{\vw_{t}, \vv_{t}} \big[ \big|
        \ell(\vy_{t}, \tsfm(\Ycal_{t-1})) - \ell(\vy_{t}, \tsfm(\Ycal'_{t-1}))
		\big| \big] 
		\leq
		\frac{2 K L_\ell \ybar}{t-\tau}.
	\end{equation*}  
\end{lemma}
\begin{proof}
From Assumption~\ref{asmp_tsfmRobustness}, it only suffices to show $\sum_{j = \tau}^{t-1} \norm{\vy_{j} - \vy'_{j}} \leq 2\ybar$.
In Assumption~\ref{asmp_tsfmRobustness}, as a result of perturbing the noise sequence $\crlbkt{\vw_{0:t}, \vv_{0:t}}$ at time $\tau$, the original sequence $\crlbkt{\vx_{0:t}, \vy_{0:t}}$ and the perturbed sequence $\crlbkt{\vx'_{0:t}, \vy'_{0:t}}$ are the same up to time $\tau - 1$ and possibly differ afterward. For $j = \tau, \cdots, t$, according to the stability in Assumption~\ref{asmp_stability}, we have $\norm{\vx_j - \vx_j'} \leq C_\rho \rho^{j - \tau} \norm{\vw_\tau - \vw_\tau'}$. 
Since, for all $i$, $g_i(\cdot)$ is assumed Lipschitz, for $j = \tau$, we have $\norm{\vy_j - \vy_j'} \leq L_g \norm{\vx_j - \vx_j'} + \norm{\vv_\tau - \vv_\tau'} \leq L_g C_\rho \rho^{j - \tau} \norm{\vw_\tau - \vw_\tau'} + \norm{\vv_\tau - \vv_\tau'}$; similarly for $j > \tau$, we have $\norm{\vy_j - \vy_j'} \leq L_g C_\rho \rho^{j - \tau} \norm{\vw_\tau - \vw_\tau'}$. Taking the summation gives that $\sum_{j = \tau}^{t-1} \norm{\vy_{j} - \vy'_{j}} \leq L_g L_\rho \norm{\vw_\tau - \vw_\tau'} + \norm{\vv_\tau - \vv_\tau'} \leq 2(L_g L_\rho \wbar + \vbar) = 2 \ybar$.
\end{proof}

\begin{proof}[Proof for Theorem~\ref{thrm_excessRisk}]
To bound the excess risk $\text{Risk}(\widehat{\tsfm}) := \Lcal(\widehat{\tsfm}) - \Lcal(\tsfm^\star)$ in \eqref{eq_excessRisk}, we first define the following empirical risk on the $M$ source systems
\begin{equation}\label{eq_empiricalRisk}
    \hat{\Lcal}(\tsfm) := \frac{1}{MT}\sum_{i=1}^M \sum_{t=1}^T  \underbrace{\ell(\vy_{i,t}, \tsfm(\Ycal_{i,t-1}))}_{=:\ell_{i,t}},
\end{equation}
Noticing that $\widehat{\tsfm} = \arg \min_{\tsfm \in \Acal} \hat{\Lcal}(\tsfm)$, the decomposition $\text{Risk}(\widehat{\tsfm}) = (\Lcal(\widehat{\tsfm}) - \hat{\Lcal}(\widehat{\tsfm})) + ( \hat{\Lcal}(\widehat{\tsfm}) - \hat{\Lcal}(\tsfm^\star)) + (\hat{\Lcal}(\tsfm^\star) - \Lcal(\tsfm^\star))$ becomes $\text{Risk}(\widehat{\tsfm}) \leq (\Lcal(\widehat{\tsfm}) - \hat{\Lcal}(\widehat{\tsfm})) + (\hat{\Lcal}(\tsfm^\star) - \Lcal(\tsfm^\star))$. This further gives
\begin{equation}
\label{eq_unnamed6}
\textstyle
    \text{Risk}(\widehat{\tsfm}) \leq 2 \sup_{\tsfm \in \Acal} |\hat{\Lcal}(\tsfm) - \Lcal(\tsfm)|.
\end{equation}
In the following, we proceed as follows: (i) assume the noise sequence $\crlbkt{\vw_{i,t}, \vv_{i,t}}$ is bounded and show that for any $\tsfm \in \Acal$, $|\hat{\Lcal}(\tsfm) - \Lcal(\tsfm)|$ is bounded; (ii) use a covering number argument to bound $\sup_{\tsfm \in \Acal} |\hat{\Lcal}(\tsfm) - \Lcal(\tsfm)|$; (iii) show $\crlbkt{\vw_{i,t}, \vv_{i,t}}$ can be bounded with high probability.

\noindent \textbf{Step (i)}: Upper bound $|\hat{\Lcal}(\tsfm) - \Lcal(\tsfm)|$

Define the following risks for the system $i = 0,1,\dots,M$
\begin{equation*}
    \textstyle
    \hat{\Lcal}_i(\tsfm) := T^\inv \sum \nolimits_{t = 1}^T \ell_{i,t},
    \quad
    \Lcal_i(\tsfm) := T^\inv \sum \nolimits_{t = 1}^T \expctn[\ell_{i,t}].
\end{equation*}
This gives $\hat{\Lcal}(\tsfm) {=} M^\inv\sum_{i=1}^M \hat{\Lcal}_i(\tsfm) $ and $\Lcal(\tsfm) {=} \Lcal_0(\tsfm)$ ${=} M^\inv\sum_{i=1}^M \Lcal_i(\tsfm)$ since $\Lcal_i(\tsfm)$ i.i.d. for all $i$. We then have $|\hat{\Lcal}(\tsfm) - \Lcal(\tsfm)| = |M^\inv \sum_{i = 1}^M (\hat{\Lcal}_i(\tsfm) - \Lcal_i(\tsfm))|$. We will bound each individual $\hat{\Lcal}_i(\tsfm) - \Lcal_i(\tsfm)$ and then apply concentration result to bound $|\hat{\Lcal}(\tsfm) - \Lcal(\tsfm)|$. 

Define the event $\Bcal_M := \cap_{i = 0}^M \cap_{j = 0}^T \crlbkt{\norm{\vw_{i,j}} \leq \wbar, \norm{\vv_{i,j}} \leq \vbar}$ for some $\wbar, \vbar \geq 0$, and Let $\prob'(\cdot) := \prob(\ \cdot \mid \Bcal_M)$ and $\expctn'[\cdot] := \expctn[\ \cdot \mid \Bcal_M]$ denote the probability measure and expectation conditioning on the event $\Bcal_M$.
Let $\Scal_{i,t} := \crlbkt{\vw_{i,0:t}, \vv_{i,0:t}}$ for $t \geq 1$ and $\Scal_{i,0} := \phi$.
Define $X_{i,t} := \expctn'[\hat{\Lcal}_i(\tsfm) \mid \Scal_{i,t}]$, then the process $\crlbkt{X_{i,t}}_{t=0}^{T}$ forms a Doob's martingale. Particularly, note that $X_{i,T} = \hat{\Lcal}_i(\tsfm)$ and $X_{i,0} = \expctn'[\hat{\Lcal}_i(\tsfm)]$. Consider the martingale difference $X_{i,\tau} - X_{i, \tau-1}$, we have
\begin{equation*}
\begin{split}
    &|X_{i,\tau} - X_{i, \tau-1}| \\
    = &\Big| T^\inv \textstyle\sum \nolimits_{t = 1}^T \expctn'[\ell_{i,t} \mid \Scal_{i,\tau}]  - \expctn'[\ell_{i,t} \mid \Scal_{i,\tau-1}] \Big| \\
    = &\Big| T^\inv \textstyle\sum \nolimits_{t = \tau}^T \expctn'[\ell_{i,t} \mid \Scal_{i,\tau}]  -  \expctn'[\ell_{i,t} \mid \Scal_{i,\tau-1}] \Big| \\
    \leq & B/T + T^\inv \textstyle\sum \nolimits_{t = \tau+1}^T \Big| \expctn'[\ell_{i,t} \mid \Scal_{i,\tau}]  - \expctn'[\ell_{i,t} \mid \Scal_{i,\tau-1}]\Big|,
\end{split}    
\end{equation*}
where the last line used the fact $\ell(\cdot, \cdot) \leq B$.
Note that each summand in the above summation can be upper bounded by $\frac{2K L_\ell \ybar}{t-\tau}$ according to Lemma~\ref{lemma_tsfmRobustness_new}. The equation above then gives
$|X_{i,\tau} - X_{i, \tau-1}| \leq T^\inv \prnthbig{B + 2K L_\ell \ybar \log(T) }$.
With this martingale difference bound, applying the Azuma-Hoeffding's inequality to $\{X_{i,t}\}_{t=0}^T$ gives
\begin{equation*}
\begin{split}
    \prob' (|X_{i,T} - X_{i,0}| \geq \epsilon)
    =&\prob (|\hat{\Lcal}_i(\tsfm) - \Lcal_i(\tsfm) - \Delta_i | \geq \epsilon  \mid \Bcal_M) \\
    \leq & 2 e^{ - \frac{T \epsilon^2}{(B + 2K L_\ell \ybar \log(T) )^2}}.
\end{split}
\end{equation*}
where $\Delta_i := \Lcal_i(\tsfm) - \expctn'[\hat{\Lcal}_i(\tsfm)]$.
Let $Y_i := \hat{\Lcal}_i(\tsfm) - \Lcal_i(\tsfm) - \Delta_i$, then the above equation tells that $Y_i$ is sub-Gaussian conditioning on $\Bcal_M$. Following from sub-Gaussian concentration bound, we have 
\begin{multline}
\label{eq_unamed7}
    \textstyle
    \prob\prnthBig{M^\inv \big| \sum \nolimits_{i = 1}^M Y_i \big| \geq \epsilon \ \Big| \Bcal_M} \leq
    e^{ - \frac{c M T \epsilon^2}{(B + 2K L_\ell \ybar \log(T))^2}},
\end{multline}
for some absolute constant $c$. This further translates to, conditioning on $\Bcal_M$, with probability at least $1-\delta$,
\begin{equation*}    
    \textstyle
    \big| M^\inv \sum \nolimits_{i = 1}^M Y_i \big|
    \leq
    \prnthbig{B + 2K L_\ell \ybar \log(T)} \sqrt{\log(2/\delta)/(c MT)}.
\end{equation*}
The definition of $Y_i$ gives $|M^\inv \sum \nolimits_{i =1}^M Y_i | = | \hat{\Lcal}(\tsfm) - \Lcal(\tsfm)$
$ - M^\inv \sum_{i =1}^M\Delta_i| \geq | \hat{\Lcal}(\tsfm) - \Lcal(\tsfm) | - |M^\inv \sum_{i =1}^M\Delta_i|$. This implies, conditioning on $\Bcal_M$, with probability at least $1-\delta$,
\begin{multline}
    \label{eq_unnamed11}
    \textstyle
    \big| \hat{\Lcal}(\tsfm) - \Lcal(\tsfm) \big|
    \leq 
    \big|M^\inv \sum_{i =1}^M\Delta_i\big| + \\
    \prnthbig{B + 2K L_\ell \ybar \log(T)} \sqrt{\log(2/\delta) / (c MT)}.
\end{multline}

\noindent \textbf{Step (ii)}: Upper bound $\sup_{\tsfm \in \Acal} |\hat{\Lcal}(\tsfm) - \Lcal(\tsfm)|$

Let $h(\tsfm):= \hat{\Lcal}(\tsfm) - \Lcal(\tsfm)$, here we seek to upper bound $\sup_{\tsfm \in \Acal} |h(\tsfm)|$. For $\epsilon>0$, let $\epsilon' := \epsilon/(\wbar + \vbar)$ and let $\Acal_{\epsilon'}$ denote the minimal $\epsilon'$-covering of $\Acal$, under the distance $\mu$ in Definition~\ref{def_algDistance}. Note that $|\Acal_{\epsilon'}| = \N(\Acal, \mu, \epsilon')$. This gives that,
\begin{equation}
\small
\label{eq_unnamed3}
    \hspace{-1em}
    \sup_{\tsfm \in \Acal} \hspace{-0.7ex} |h(\tsfm)| {\leq} \max_{\tsfm \in \Acal_{\epsilon'}} \hspace{-0.7ex} |h(\tsfm)| {+} \sup_{\tsfm \in \Acal} \min_{\tsfm' \in \Acal_{\epsilon'}} \hspace{-0.7ex} |h(\tsfm) {-} h(\tsfm')|.
\end{equation}

For the term $\max_{\tsfm \in \Acal_{\epsilon'}} |h(\tsfm)|$, applying the union bound to \eqref{eq_unnamed11} for all $\tsfm \in \Acal_{\epsilon'}$, we obtain that conditioning on $\Bcal_M$, with probability at least $1-\delta$,
\begin{multline}
\label{eq_unnamed4}
    \textstyle
    \hspace{-1em}
    \max_{\tsfm \in \Acal_{\epsilon'}} |h(\tsfm)| 
    \leq 
    \big|M^\inv \sum_{i = 1}^M \Delta_i\big| + \\
    \prnthbig{B {+} 2K L_\ell \ybar \log(T)} \sqrt{\log(2 \N(\Acal, \mu, \epsilon')/\delta)/ (c MT)}.
\end{multline}

Next we consider the term $\sup_{\tsfm \in \Acal} \min_{\tsfm' \in \Acal_{\epsilon'}} |h(\tsfm) - h(\tsfm')|$ in \eqref{eq_unnamed3}. 
Let $\Lcal'(\tsfm) := \expctn[\hat{\Lcal}_0(\tsfm) \mid \Bcal_M]$, $\Delta_{h,1} := | \hat{\Lcal}(\tsfm) - \hat{\Lcal}(\tsfm') |$, $\Delta_{h,2} := |\Lcal'(\tsfm) - \Lcal'(\tsfm')|$, and $\Delta_{h,3} := |\Lcal(\tsfm) - \Lcal'(\tsfm)| + |\Lcal(\tsfm') - \Lcal'(\tsfm')|$. 
Using the definition of $h(\cdot)$ and the triangular inequality, we have
$|h(\tsfm) - h(\tsfm')| \leq \Delta_{h,1} + \Delta_{h,2} + \Delta_{h,3}$. By the Lipschitzness of the loss function $\ell$ and the bound on the distance between $\tsfm$ and $\tsfm'$, i.e., $\mu(\tsfm, \tsfm') \leq \epsilon/(\wbar + \vbar)$, we obtain that, conditioning on $\Bcal_M$, both $\Delta_{h,1}$, $\Delta_{h,2}$ can be upper bounded by $L_\ell \epsilon$. This gives 
\begin{equation}
\label{eq_unnamed12}
    \sup_{\tsfm \in \Acal} \min_{\tsfm' \in \Acal_{\epsilon'}} |h(\tsfm) - h(\tsfm')| \leq 2L_\ell \epsilon + \Delta_{h,3}.
\end{equation}

Plugging \eqref{eq_unnamed12} and \eqref{eq_unnamed4} into \eqref{eq_unnamed3} followed by invoking $\text{Risk}(\widehat{\tsfm})  \leq 2 \sup_{\tsfm \in \Acal} |h(\tsfm)|$ in \eqref{eq_unnamed6} gives that, conditioning on $\Bcal_M$, with probability at least $1-\delta$
\begin{multline}
\label{eq_unnamed5}
    \hspace{-1.5em}
    \small
    \textstyle
    \text{Risk}(\widehat{\tsfm})
    \leq 
    4 L_\ell \epsilon +2\Delta_{h,3} + 2\big|M^\inv \sum_{i = 1}^M \Delta_i\big| + \\
    \hspace{-0.2em}
    2\prnthbig{B {+} 2K L_\ell \ybar \log(T)} \sqrt{\log(2 \N(\Acal, \mu, \epsilon')/\delta)/ (c MT)}.
\end{multline}

\noindent \textbf{Step (iii)}: Upper bound the noise sequence $\crlbkt{\vw_{i,t}, \vv_{i,t}}$

Let $\Ecal$ denote the event in \eqref{eq_unnamed5}, then we have $\prob(\Ecal \mid \Bcal_M) \geq 1-\delta$. In the event $\Bcal_M$, we now set $\wbar = \sqrt{3} \sigma_\vw \sqrt{\log(2MT/\delta)}$ and $\vbar = \sqrt{3} \sigma_\vv \sqrt{\log(2MT/\delta)}$. 
Using the Gaussian concentration bound and the union bound, we obtain that $\prob(\Bcal_M) \geq 1 - \delta$, when $MT \geq 3 \max(\sqrt{n}, \sqrt{m})$. This further yields
\begin{equation}
\label{eq_unnamed8}
\hspace{-1.5em}
    \prob(\Ecal) \geq \prob(\Ecal, \Bcal_M) \geq \prob(\Ecal |\Bcal_M) \prob(\Bcal_M) \geq (1{-}\delta)^2 \geq 1{-}2\delta.
\end{equation}
Now we inspect the term $|M^\inv \sum_{i = 1}^M\Delta_i|$ in the definition of $\Ecal$, i.e., \eqref{eq_unnamed5}. Note that by definition 
$|\Delta_i| 
    = |\expctn[\hat{\Lcal}_i(\tsfm)] - \expctn[\hat{\Lcal}_i(\tsfm) | \Bcal_M]| 
    \leq |\expctn[\hat{\Lcal}_i(\tsfm) | \Bcal_M] (\prob(\Bcal_M)-1)| + |\expctn[\hat{\Lcal}_i(\tsfm) | \Bcal_M^c] \prob(\Bcal_M^c)| 
    \leq 2 B \delta$, 
where the facts $\ell(\cdot, \cdot) \leq B$ and complement probability $\prob(\Bcal_M^c) \leq \delta$ are used. Hence, we have $|M^\inv \sum_{i = 1}^M\Delta_i| \leq 2 B \delta$. 
Similarly, we can show $\Delta_{h,3} \leq 4B\delta$. With these bounds, invoking \eqref{eq_unnamed8} gives, with probability at least $1-2\delta$,
\begin{multline*}
    \text{Risk}(\widehat{\tsfm})
    \leq 
    4L_\ell \epsilon  + 12 B \delta + \\
    2\prnthbig{B + 2K L_\ell \ybar \log(T)} \sqrt{\log(2 \N(\Acal, \mu, \epsilon')/\delta)/(c MT)}.
\end{multline*}
Finally, plugging in the definition $\ybar:=L_g L_\rho \wbar + \vbar$ and $\epsilon' := \epsilon/(\wbar + \vbar)$ concludes the proof.
\end{proof}

\section{Systems that are hard to learn in-context}
\label{sec_hardsystem}

In this section, we investigate two limitations of \MF, one explained by our theoretical guarantees, the other regarding the performance degradation in the face of distribution shifts. 

To illustrate the first challenge, consider two distinct classes of linear systems. The first class employs the same generation procedure as described in Section \ref{subsec:lin}. In the second class, we follow a similar generation procedure, except for the $\vA$ matrices, which are generated as upper-triangular matrices. Here, the diagonal entries are sampled from the interval $[-0.95, 0.95]$, while the upper triangular entries are sampled from the range $[-1, 1]$. The experimental results, presented in Fig.~\ref{fig_hardSys}, demonstrate that, compared with the densely generated $\vA$ matrices, the upper-triangular $\vA$ matrices make it harder for \MF to learn the optimal Kalman filter. As depicted in Fig.~\ref{fig_hardSys}, the powers of the upper-triangular $\vA$ matrices exhibit a slower decay rate and even initial overshoot in comparison to those of the dense $\vA$ matrices. 
Noticing that $\vy_t = \sum_{i=1}^t \vC \vA^i \vw_{t-i} + \vv_t$, this implies that upper triangular $\vA$ establishes stronger and longer temporal correlation between $\vy_t$ and past $\vy$'s, i.e., slow mixing. This poses challenges to \MF but can be potentially mitigated by feeding \MF longer prompts, i.e. the time horizon $T$. Theoretically, the slow decay rate implies larger $L_\rho := C_\rho/(1-\rho)$, which consequentially gives a looser risk upper bound in Theorem \ref{thrm_excessRisk}.

In our experiments in Section~\ref{sec_experiments}, the distribution the source and target systems are drawn from is the same. Here we run an experiment to illustrate how \MF behaves if the target distribution is different than the source one. In particular, under the experimental setup of Section~\ref{subsec:lin}, we train the \MF with noise covariances $\sigma_{\vw}^2 =  \sigma_{\vv}^2 =0.1 \vI_n$ and test on systems subject to a different noise covariance. As shown in Fig.~\ref{fig_dist_shift}, \MF's performance degrades when the target systems are subject to a different noise distribution, especially when the noise covariance increases.

\begin{figure}
    \noindent
    \hspace{-1.2em}
    \centering
    \includegraphics[scale=0.34]{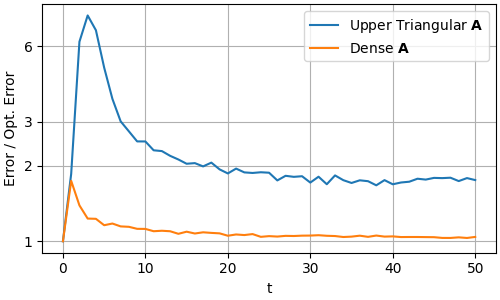} 
    \includegraphics[scale=0.34]{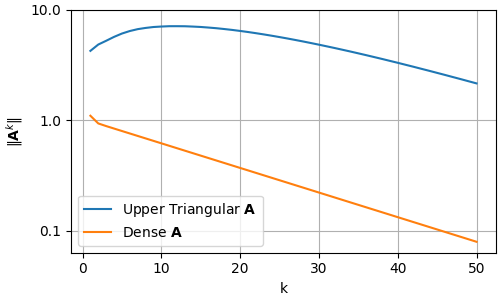}  
    \footnotesize \customlabelfig{fig_errorratio}{\ref{fig_hardSys}(a)} (a) \hspace{0.45 \linewidth}
    \customlabelfig{fig_matrixpower}{\ref{fig_hardSys}(b)} (b)
    \vspace{-1em}
    \caption{Comparison between dense and upper-triangular $\vA$ matrices: (a) prediction error ratio between \MF and Kalman filter; (b) matrix powers averaged over all source systems.}
    \label{fig_hardSys} 
\end{figure}

\begin{figure}
\centering
    \includegraphics[scale=0.45]{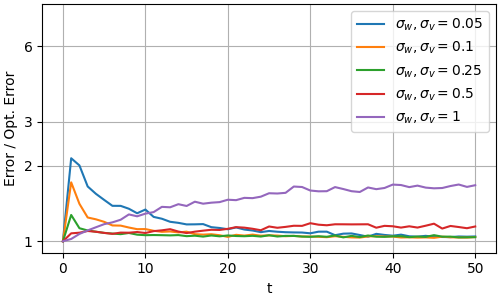} 
    \vspace{-1em}
    \caption{Performance of \MF compared to Kalman Filter when noise level in test is different than train.}
    \label{fig_dist_shift}  
    \vspace{-0.5cm}
\end{figure}

\section{Conclusion}

In conclusion, this work has demonstrated the potential of transformers in addressing prediction problems for dynamical systems. The proposed \MF exhibits remarkable performance by adapting to unseen settings, non-i.i.d. noise, and time-varying dynamics.

This work motivates new avenues for the application of transformers in continuous control and dynamical systems. Future work could extend the \MF approach to closed-loop control problems to meta-learn policies for problems such as the optimal quadratic control. It is also of interest to explore new training strategies to promote robustness (e.g., against distribution shifts) and safety of this approach in control problems.

\section{ACKNOWLEDGMENTS}
The authors would like to thank Sultan Daniels, Gautam Goel, Wentinn Liao, Gireeja Ranade, and Anant Sahai for their careful efforts on reproducing the experimental results and their feedback, which helped us revise the descriptions of the experimental setup.

\bibliographystyle{IEEEtran}
{
\small
\bibliography{citations.bib}
}

\end{document}

%% file: preamble.tex
\usepackage[english]{babel}
\usepackage[utf8]{inputenc}
\usepackage{amsmath}
\usepackage{amsfonts}
\usepackage{comment}
\usepackage{balance}

\usepackage{graphicx}
\graphicspath{ {images/}}

\usepackage{pifont}

\makeatletter
\let\NAT@parse\undefined
\makeatother
\usepackage[colorlinks,urlcolor=blue,linkcolor=blue,citecolor=blue]{hyperref}

\usepackage[ruled, vlined]{algorithm2e}

\usepackage{xstring}

\usepackage{cite}

\usepackage{multirow}
\usepackage{calc}
\usepackage{adjustbox}
\usepackage{booktabs}

\usepackage{makecell}

\usepackage{bbm}

\newtheorem{theorem}{Theorem}
\newtheorem{lemma}{Lemma}

\newtheorem{definition}{Definition}
\newtheorem{assumption}{Assumption}

\let\labelindent\relax
\usepackage{enumitem}

\allowdisplaybreaks

\input{preambleLettersDef.tex}

\newcommand{\norm}[1]{\|{#1} \|}

\newcommand{\crlbkt}[1]{\{ #1 \}}

\newcommand{\sqbktbig}[1]{\bigl[ #1 \bigr]}

\newcommand{\prnthbig}[1]{\bigl( #1 \bigr)}
\newcommand{\prnthBig}[1]{\Bigl( #1 \Bigr)}

\newcommand{\R}{\mathbb{R}}				
\newcommand{\dm}[2]
{
	\IfStrEq{#2}{1}{\R^{#1}}{\R^{#1 \x #2}}
}

\newcommand{\N}{\mathcal{N}}			

\newcommand{\expctn}{\mathbb{E}} 	 	
\newcommand{\inv}{{-1}} 				
\newcommand{\prob}{\mathbb{P}} 				

\usepackage{amssymb}
\makeatletter
\newcommand*{\T}{{\mathpalette\@transpose{}}}
\newcommand*{\@transpose}[2]{\raisebox{\depth}{$\m@th#1\intercal$}}
\makeatother

\newcommand*{\x}{\mathsf{x}\mskip1mu} 	

\newcommand{\splitatcommas}[1]{%
	\begingroup
	\begingroup\lccode`~=`, \lowercase{\endgroup
		\edef~{\mathchar\the\mathcode`, \penalty0 \noexpand\hspace{0pt plus .1em}}%
	}\mathcode`,="8000 #1%
	\endgroup
}

\newcommand\Item[1][]{%
\ifx\relax#1\relax  \item \else \item[#1] \fi
\abovedisplayskip=0pt\abovedisplayshortskip=0pt~\vspace*{-\baselineskip}}

\newcounter{reviewer}
\newcounter{point}[reviewer]
\makeatletter
\newcommand{\customlabelfig}[2]{%
    \protected@write \@auxout {}{
    \string \newlabel {#1}{{#2}{\thepage}{#2}{#1}{}{}}
    }
}
\makeatother

\makeatletter
\newcommand{\customlabelresponse}[2]{%
   \protected@write \@auxout {}{
   \string \newlabel {#1}{{#2}{\thepage}{#2}{#1}{}} 
   }%
   \hypertarget{#1}{#2}
}
\makeatother



%% file: preambleLettersDef.tex
\newcommand{\bv}[1]{\mathbf{#1}}		
\newcommand{\bvgrk}[1]{{\boldsymbol{#1}}}	

\newcommand{\jbar}{\bar{j \phantom{\tiny \,}} \kern -0.1em}

\newcommand{\vbar}{\bar{v}}
\newcommand{\wbar}{\bar{w}}

\newcommand{\ybar}{\bar{y}}

\newcommand{\jhat}{\hat{j \phantom{\tiny \,}} \kern -0.1em}

\newcommand{\jtil}{\tilde{j \phantom{\tiny \,}} \kern -0.1em}

\newcommand{\vj}{\bv{j}}

\newcommand{\vv}{\bv{v}}
\newcommand{\vw}{\bv{w}}
\newcommand{\vx}{\bv{x}}
\newcommand{\vy}{\bv{y}}

\newcommand{\vjbar}{\bar{\vj \phantom{\tiny \,}} \kern -0.1em}

\newcommand{\vjhat}{\hat{\vj \phantom{\tiny \,}} \kern -0.1em}

\newcommand{\vyhat}{\hat{\vy}}

\newcommand{\vjtil}{\tilde{\vj \phantom{\tiny \,}} \kern -0.1em}

\newcommand{\Bbar}{\bar{B}}

\newcommand{\Acal}{\mathcal{A}}
\newcommand{\Bcal}{\mathcal{B}}

\newcommand{\Dcal}{\mathcal{D}}
\newcommand{\Ecal}{\mathcal{E}}

\newcommand{\Lcal}{\mathcal{L}}

\newcommand{\Ocal}{\mathcal{O}}

\newcommand{\Qcal}{\mathcal{Q}}

\newcommand{\Scal}{\mathcal{S}}

\newcommand{\Wcal}{\mathcal{W}}

\newcommand{\Ycal}{\mathcal{Y}}

\newcommand{\vA}{\bv{A}}

\newcommand{\vC}{\bv{C}}

\newcommand{\vI}{\bv{I}}

\usepackage{bm}

\newcommand{\vAcalhat}{\hat{\bm{\Acal} \phantom{a}} \kern-0.5em}
\newcommand{\vAcalcheck}{\check{\bm{\Acal} \phantom{a}} \kern-0.5em}
\newcommand{\vAcalbar}{\bar{\bm{\Acal} \phantom{a}} \kern-0.5em}

\usepackage{upgreek}
\newcommand{\valpha}{\bvgrk{\upalpha}}

\newcommand{\vgamma}{\bvgrk{\upgamma}}

\newcommand{\veta}{\bvgrk{\upeta}}

